\documentclass[letterpaper, 10 pt, conference]{ieeeconf}  

\IEEEoverridecommandlockouts                              

\usepackage{graphicx} 
\usepackage{amsmath} 
\usepackage{amssymb}  
\usepackage{cite}

\usepackage{color}
\newcommand{\rev}[1]{{\color{black}{#1}}}

\newtheorem{thm}{Theorem}
\newtheorem{prob}{Problem}

\newtheorem{rem}{Remark}

\DeclareMathOperator*{\minimize}{minimize}

\title{\LARGE \bf
Optimal shaping filter design for \\
data-driven feedforward controller tuning
}

\author{Yusuke Fujimoto$^{1}$
\thanks{*This work was supported by  JSPS KAKENHI Grant number 25K01254 and 26K07552. }
\thanks{$^{1}$Yusuke Fujimoto is with Graduate School of Engineering Science, the University of Osaka, Toyonaka, Osaka 560-8531, Japan. 
        {\tt\small fujimoto.yusuke.es@osaka-u.ac.jp}}%
}

\begin{document}

\maketitle
\thispagestyle{empty}
\pagestyle{empty}

\begin{abstract}

This paper discusses the data-driven model matching problem. 
In particular, this paper focuses on two-degree-of-freedom control systems, 
and consider to design feedforward controller from input-output data. 
An intuitive solution to this problem would be identifying the optimal controller using data, 
but this does not give the exact solution to the original problem. 
A shaping filter is required to compensate for this gap, and 
the main contribution of this paper is to give the optimal shaping filter. 
The proposed shaping filter is constructed from available information under reasonable assumptions, 
and its effectiveness is shown through a numerical example and a practical experiment. 
The relation between the proposed shaping filter and Estimated Response Iterative Tuning (ERIT) is 
also discussed, and it is shown that ERIT is optimal for a special case. 
\end{abstract}

\section{Introduction} \label{sec:intro}

Although mathematical models are important in model-based controller design, 
modeling itself is a difficult task in practice. 
Since it does not suffer from modeling effort and modeling error, 
data-driven approaches that design the controller directly from data 
have attracted many attentions in these days \cite{Hou:2013,Bazanella:2011}, 
and many works including practical applications have been reported 
\cite{Previdi:2010,Rojas:2010,Kano:2011,Wakasa:2012}. 

One of the most well-investigated data-driven design problems would be 
the data-driven model matching from one-shot data. 
Many methods have been proposed to solve this problem, 
e.g., Virtual Reference Feedback Tuning (VRFT) \cite{Campi:2002}, 
Fictitious Reference Iterative Tuning (FRIT) \cite{Kaneko:2013} (originally proposed in Japanese \cite{Soma:2004Eng}), 
inverse VRFT \cite{Sala:2005}, Non-iterative Data-driven Model matching (NDM) \cite{Karimi:2007}, 
Optimal Controller Identification (OCI) \cite{Campestrini:2017}, 
Virtual Internal Model Tuning (VIMT) \cite{Ikezaki:2019}, and so on. 
The above methods tune the parameters of \textit{feedback controller} from 
input-output data measured in a preliminary experiment. 
Another approach to solve the data-driven model matching problem is to design the parameter of 
\textit{feedforward controller} in a two-degrees-of-freedom control system. 
An important advantage of focusing on the feedforward controller is that the stability analysis becomes much easier than the 
one for the feedback case \cite{Hausden:2011}; if the initial feedback controller stabilizes the closed-loop and the updated feedforward controller itself is stable, the updated closed-loop system becomes stable.  
An example of the data-driven feedforward design methods is Estimated Response Iterative Tuning 
\cite{Kaneko:2017,Kaneko:2018}. 
Because it has the aforementioned advantage, 
several works related to ERIT have been reported \cite{Fujimoto:2018-3,
Fujimoto:2022,Ishihara:2023}.

Most of the above methods share a common structure; 
their cost functions can be understood as an identification problem 
whose true system is the controller which achieves perfect model matching \cite{Fujimoto:2023}. 
However, this leads to an undesirable problem. 
Let us call the controller which achieves perfect model matching the optimal controller. 
The problem is that similarity between the optimal controller and 
another controller does not immediately mean the similarity between their control performances. 
\rev{
This implies that approximating the optimal controller does not solve the original data-driven model matching problem 
even when noise-free data is available 
\cite{Campi:2002}. 
} 
A solution to this problem is to introduce a shaping filter (also known as a prefilter) to the cost function. 
An appropriate design of a shaping filter reduces the gap between the 
performance \rev{the reference model and the updated system}. 
However, the optimal shaping filter sometimes requires knowledge of the true system 
\cite{Campi:2002}, thus, how to design the shaping filter itself is an important work \cite{Matsui:2016, Kajiwara:2017}. 

Based on these backgrounds, this paper discusses the shaping filter design for 
data-driven model matching with a feedforward controller. 
\rev{We consider a parameter tuning problem of a fixed controller and noise-free data. 
} 
This problem is almost the same as the identification of an inverse model \cite{Blanken:2020,Ho:2018,Jung:2013}. 
\rev{
We mainly consider the case where perfect model matching can not be achieved with the structure. 
This is because the system structure itself is assumed to be unknown. 
One of the important applications of such a situation would be a servo motor; 
we can not know the inertia or load to be mounted on the motor in advance. 
}
The main contribution of this paper is to show the optimal shaping filter for such a feedforward identification problem. 
One of the interesting properties of the proposed shaping filter is that it can be 
constructed from available information under reasonable assumptions. 
In more detail, it only requires 
1) the initial input-output data is measured with a two-degrees-of-freedom control system setup, 
2) the reference model, feedback controller, and initial feedforward controller are known, and 
3) the spectrum of the initial reference signal is known. 
The third point is also reasonable as discussed later.

This paper is constructed as follows. 
Sec.~\ref{sec:Probset} first sets the problem discussed in this paper. 
Sec.~\ref{sec:main} gives the optimal shaping filter, and 
then Sec.~\ref{sec:relationwithERIT} discusses about the relation between the 
proposed shaping filter and ERIT. 
In particular, it is shown that ERIT is optimal for a certain setup. 
\rev{Sec.~\ref{sec:sim} and Sec.~\ref{sec:exp} show a numerical example and a practical experiment, respectively. }

\noindent
[Notation] 
The imaginary unit is denoted by $j$ throughout the paper. 
The complex frequency in the $z$-transform is denoted \rev{by $z$. }
In this paper, we use $\| G \|$ to show the $H_2$ norm of a 
single-input-single-output transfer function $G$, i.e., 
if the impulse response of $G$ is denoted by $g_k$, 
$\|G\|^2 = \sum_{k=0}^{\infty} g_k^2$. 
Throughout the paper, we slightly abuse the notation and 
regard the transfer function as both an operator and a complex function; 
for instance, the equation $y_k=Gu_k$ uses $G$ as an operator, and 
indicates that $y_k$ is the output of $G$ whose input is $u_k$. 
This implies that we use $\frac{1}{z}$ as backward shift operator, i.e., 
$\frac{1}{z} u_k = u_{k-1}$. 
We also use $G$ to denote the $z$-transform of $g_k$, and thus 
$G(e^{j \omega})$ implies that we substitute $e^{j\omega}$ to the complex function $G$. 
Based on these, it should be noted that the $H_2$ norm also has an expression given by 
$\|G\|^2=\frac{1}{2\pi} \int_{-\pi}^{\pi} \left|G(e^{j\omega})\right|^2d\omega$. 
\rev{
Another important property of the $H_2$ norm is about the 
variance of the output where the input is a white random sequence whose mean and variance are 0 and 1, respectively. 
Let $\varepsilon_k$ be the white random sequence whose mean and variance are 0 and 1, respectively. 
Then, 
$\mathbb{E}[y_k^2]=\mathbb{E}\left[\left(\sum_{i=0}^{\infty}g_i \varepsilon_{k-i} \right)^2\right]=
\sum_{i=0}^{\infty}g_i^2 =\|G\|^2$. 
}

\section{Problem setting} \label{sec:Probset}

\subsection{Data-driven model matching problem}
\begin{figure}[!t]
\centering
\includegraphics[width=7cm]{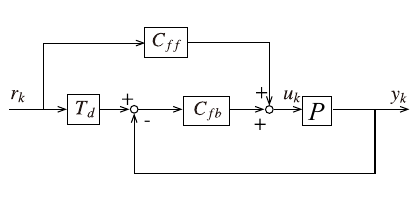}
\caption{Two-degrees of freedom control system}
\label{fig:2dof}
\end{figure}
This work focuses on 
a two-degrees-of-freedom control system illustrated in Fig.~\ref{fig:2dof}. 
The reference signal, control input, and output at time step $k$ are  
denoted by $r_k, u_k, $ and $y_k$, respectively. 
We consider a noise-free setting in this paper. 
The target system is a linear time-invariant discrete-time system 
whose transfer function is given by $P$. 
The feedback controller, feedforward controller, and reference model are 
denoted by $C_{fb}, C_{ff}(\rho)$, and $T_d$, respectively. 
Here, $\rho$ denotes the parameter of the feedforward controller. 
\rev{We assume that 
feedback loop is internally stable with $C_{fb}$. 
} 
The closed-loop system is denoted by $T(\rho)$, i.e., 
\begin{align}
T(\rho)=\frac{P(C_{ff}(\rho)+T_d C_{fb})}{1+PC_{fb}}. 
\end{align}

Based on the above notation, the data-driven model matching problem is given as follows. 
\begin{prob} \label{prob:DDCFF}
Let $W$ be the weight function. 
Assume that $P$ is unknown. 
Find $\rho$ which minimizes 
$J(\rho)=\| W( T(\rho) -T_d )\|^2$ from 
$W, T_d, C_{fb}$ and data $r_k^0, u_k^0, y_k^0 (k=0, \ldots, N-1)$ 
measured in a preliminary experiment with the initial feedforward parameter $\rho_0$.
\end{prob}


\subsection{Shaping filter design}

It is well recognized that $C^*_{ff}=\frac{T_d}{P}$ exactly makes the closed-loop in Fig.~\ref{fig:2dof} 
equivalent to $T_d$. 
Let $\mathcal{C}$ be the set of feasible feedforward controllers, i.e., 
\begin{align}
\mathcal{C}=\left\{C\mid  C=C_{ff}(\rho) \right\}. 
\end{align}
If $C^* \in \mathcal{C}$, i.e., if there exists $\rho^*$ which satisfies $C_{ff}(\rho^*)=C^*$,  
\rev{the optimal parameter $\rho^*$ is the minimizer of the following cost function.} 
\begin{align}
J_0(\rho)=\sum_{k=0}^{N-1} \left(T_d u_k^0 - C_{ff}(\rho)y_k^0 \right)^2
\end{align}
This is almost trivial by noting $C_{ff}(\rho^*)=\frac{T_d}{P}$ and $y_k^0 = Pu_k^0$. 
In fact, 
\begin{align}
J_0(\rho^*)=&\sum_{k=0}^{N-1} \left(T_d u_k^0 - C_{ff}(\rho^*)y_k^0 \right)^2 \nonumber \\ 
=&\sum_{k=0}^{N-1} \left(T_d u_k^0 - \frac{T_d}{P}Pu_k^0 \right)^2=0,
\end{align}
thus $\rho^*$ is one of the minimizer of $J_0(\rho)$. 
When $y_k^0$ is sufficiently excited, the minimization of $J_0(\rho)$ gives $\rho^*$. 

This paper focuses on the case where $C^* \notin \mathcal{C}$. 
The minimization of $J_0(\rho)$ does not give the exact solution of Problem~\ref{prob:DDCFF} in this case.  
Now the shaping filter design problem which is the main topic of this paper is formulated as follows. 
\begin{prob} \label{prob:shapingfilter}
\rev{Assume that the same information as Problem~\ref{prob:DDCFF} is available.}
Consider a cost function 
\begin{align}
J_L(\rho) = \sum_{k=0}^{N-1} \Bigl( L\left(T_d u_k^0 - C_{ff}(\rho)y_k^0 \right)\Bigr)^2 \label{eq:JL}
\end{align}
where $L$ is a shaping filter. 
Design $L$ so that minimizing $J_L(\rho)$ is equivalent to minimizing $J(\rho)$. 
\end{prob}

We also make an assumption on the reference signal. 
Let $\delta_k$ be the impulse signal given by 
\begin{align}
\delta_k = \begin{cases}
1 & \rev{k=0} \\
0& \rev{k\neq 0}
\end{cases}. 
\end{align} 
We assume the following \rev{condition} for the reference signal $r_k$. 
\begin{description}
\item[(A1)] 
The reference signal $r_k$ is described with 
a known filter $R$. In more detail, $r_k$ is one of the following signals. 
\begin{enumerate}
\item $r_k$ is given by 
$r_k = R \delta_k $. A typical example is the step signal, i.e., 
$R$ is the discrete integrator. 
\item $r_k$ is given by 
$r_k=R \varepsilon_k$ where $\varepsilon_k$ is \rev{a white noise} whose mean and variances are 0 and 1, respectively. 
\end{enumerate}
\end{description}
The assumption (A1) implies that the spectrum of $r_k$ is known. 
By noting $R$ can be $R=1$, (A1) is reasonable since 
this includes several widely-used references, e.g., step signal, multi-sine signal, 
white noise, band-limited noise, etc. 

\section{Main result} \label{sec:main}

Now the main result of the paper is summarized as follows. 
\begin{thm} \label{thm:main}
Let 
\begin{align}
L=\frac{1}{C_{ff}(\rho_0)+T_d C_{fb}}\frac{W}{R}\frac{1}{z^m} \label{eq:optimalL}
\end{align}
where 
$\frac{1}{z^m}$ is pure delay, which is introduced to make the filter proper.  
If this $L$ is stable, then minimizing $J_L(\rho)$ becomes equivalent to minimizing $J(\rho)$ with $N\to \infty$ with this shaping filter under the assumption (A1). 
\end{thm}
\begin{proof}
First of all, note that 
the ideal cost function $J(\rho)$ satisfies 
\begin{align}
J(\rho)=&\|W(T(\rho)-T_d)\|^2 \nonumber\\
=&\left\|W\left( \frac{P(C_{ff}(\rho)+T_d C_{fb})}{1+PC_{fb}}-T_d\frac{1+PC_{fb}}{1+PC_{fb}}\right) \right\|^2\nonumber \\
=&\left\| \frac{W}{1+PC_{fb}}(C_{ff}(\rho)P-T_d) \right\|^2. \label{eq:Jrho_transformed}
\end{align}

Consider case 1 of (A1), i.e., $r_k=R\delta_k$. 
In this case, we have 
\begin{align}
\begin{split}
u^0_k &= \frac{C_{ff}(\rho_0)+T_d C_{fb}}{1+PC_{fb}}R \delta_k, \\ 
y^0_k &= \frac{P(C_{ff}(\rho_0)+T_d C_{fb})}{1+PC_{fb}}R \delta_k, 
\end{split}
\end{align}
thus if we employ $L$ defined by \eqref{eq:optimalL}, 
we have 
\begin{align}
Lu^0_k = \frac{W}{1+PC_{fb}} \delta_{k-m}, 
Ly^0_k = \frac{PW}{1+PC_{fb}} \delta_{k-m}.  
\end{align}
Based on these, 
\begin{align}
J_L(\rho)=\sum_{k=0}^{N-1}\left(\frac{W}{1+PC_{fb}}\left(
PC_{ff}(\rho)-T_d
\right) \delta_{k-m} \right)^2. 
\end{align}
Since this shows the squared sum of the impulse response of 
$\frac{W}{1+PC_{fb}}\left(
PC_{ff}(\rho)-T_d
\right)$, $J_L(\rho)$ converges to \rev{$J(\rho)$} with $N\to \infty$, and the 
statement has been proven. 

Now consider case 2, i.e., $r_k =R \varepsilon_k$. 
Note that minimizing $J_L(\rho)$ is equivalent to 
\begin{align}
\bar{J}_L(\rho) =\frac{1}{N}\sum_{k=0}^{N-1} \Bigl( L\left(T_d u_k^0 - C_{ff}(\rho)y_k^0 \right)\Bigr)^2.
\end{align}
In the case with $r_k = R\varepsilon_k$ and \eqref{eq:optimalL}, 
\begin{align}
\bar{J}_L(\rho) = 
\frac{1}{N} \sum_{k=0}^{N-1}\Bigl( \frac{W}{1+PC_{fb}} \left(T_d  - C_{ff}(\rho)P \right)\varepsilon_{k-m}\Bigr)^2. 
\end{align}
Here we omit some transformations which are similar to the former discussions. 
Since $\varepsilon_k$ is assumed to be white noise with variance 1, 
$\bar{J}_L(\rho) 
$ converges 
\rev{the sum of squared error of the impulse response of $
\frac{W}{1+PC_{fb}} \left(T_d  - C_{ff}(\rho)P \right)$, 
i.e., $\left\|\frac{W}{1+PC_{fb}} \left(T_d  - C_{ff}(\rho)P \right) \right\|^2$ 
in probability 
with $N\to \infty$. 
This means that the optimal solution of 
$\bar{J}_L(\rho)$ converges to 
the one of $J(\rho)$. 
}
This completes the proof. 
\end{proof}

From the proof, we have another expression of the result; 
the optimal shaping filter satisfies
\begin{align}
|L|^2= \left|\frac{1}{C_{ff}(\rho_0)+T_dC_{fb}}\frac{W}{R} \right|^2, 
\end{align}
where the argument $e^{j\omega}$ is omitted to make the notation easy. 

Two points should be noted; 
first, the optimal filter is not unique. 
In fact, multiplying an all-pass filter (such as $z^{-1}$) does not change 
the $H_2$ norm, thus a filter $L$ multiplied by an all-pass filter is another 
optimal filter. 
This observation also gives an important suggestion; 
if $L$ is unstable, inner-outer decomposition provides the optimal shaping filter. 
Second, $L$ given by \eqref{eq:optimalL} is constructed from available information; 
$C_{ff}(\rho_0), T_d, C_{fb}$ and $W$. 
Hence the optimal filter is available under \rev{reasonable assumptions}. 

\section{Relation with ERIT} \label{sec:relationwithERIT}

\subsection{Brief introduction of ERIT}

Estimated Response Iterative Tuning (ERIT) is one of the solutions \rev{to} Problem~\ref{prob:DDCFF}. 
This section discusses the relationship between the proposed optimal shaping filter and ERIT. 
To this end, we first briefly introduce the idea of ERIT. 

Consider the initial output $y_k^0$ with the initial parameter 
$\rho_0$. 
Since the closed-loop in Fig.~\ref{fig:2dof} is considered, 
\begin{align}
y_k^0 = \frac{P(C_{ff}(\rho_0)+T_d C_{fb})}{1+PC_{fb}}r_k \label{eq:yk0}
\end{align}
holds. 
Now consider the output with parameter $\rho$ denoted by $y_k(\rho)$. 
This signal satisfies 
\begin{align}
y_k(\rho) = \frac{P(C_{ff}(\rho)+T_d C_{fb})}{1+PC_{fb}}r_k. \label{eq:ykrho}
\end{align}
From \eqref{eq:yk0} and \eqref{eq:ykrho}, it holds 
\begin{align}
y_k(\rho) = \frac{C_{ff}(\rho)+T_d C_{fb}}{C_{ff}(\rho_0)+T_d C_{fb}} y_k^0. 
\end{align}
Since $T(\rho)$ becomes equivalent to $T_d$ when $y_k(\rho)=T_d r_k$, 
ERIT employs the cost function defined as 
\begin{align}
J_{erit}(\rho) = \sum_{k=0}^{N-1} \left(T_d r_k-  \frac{C_{ff}(\rho)+T_d C_{fb}}{C_{ff}(\rho_0)+T_d C_{fb}} y_k^0\right)^2,  \label{eq:Jerit}
\end{align}
and selects the parameter which minimizes $J_{erit}(\rho)$. 

\subsection{Optimality of ERIT}

Now consider the cost function of ERIT given by \eqref{eq:Jerit}. 
By noting 
\begin{align}
T_dr_k = T_d \frac{1+PC_{fb}}{1+PC_{fb}} r_k, 
\end{align}
and 
\begin{align}
&\frac{T_d C_{fb}}{C_{ff}(\rho_0)+T_d C_{fb}} y_k^0 \nonumber \\
&\quad =\frac{T_d C_{fb}}{C_{ff}(\rho_0)+T_d C_{fb}}
\frac{P(C_{ff}(\rho_0)+T_d C_{fb})}{1+PC_{fb}}r_k \nonumber\\
&\quad =\frac{T_d  PC_{fb}}{1+PC_{fb}}r_k,
\end{align}
\rev{it holds that}
\begin{align}
&T_d r_k-  \frac{C_{ff}(\rho)+T_d C_{fb}}{C_{ff}(\rho_0)+T_d C_{fb}} y_k^0 \nonumber \\
&\quad =  \frac{T_d}{1+PC_{fb}}r_k-\frac{C_{ff}(\rho)}{C_{ff}(\rho_0)+T_d C_{fb}} y_k^0. \label{eq:ErrorInERIT}
\end{align}
Since the initial input $u_k^0$ satisfies 
\begin{align}
u_k^0 = \frac{C_{ff}(\rho_0)+T_d C_{fb}}{1+PC_{fb}}r_k, 
\end{align}
\eqref{eq:ErrorInERIT} suggests 
\begin{align}
&T_d r_k-  \frac{C_{ff}(\rho)+T_d C_{fb}}{C_{ff}(\rho_0)+T_d C_{fb}} y_k^0 \nonumber \\
&\quad =  \frac{1}{C_{ff}(\rho_0)+T_d C_{fb}} (T_d u_k-C_{ff}(\rho)y_k^0). 
\end{align}

Based on the above discussions, the cost function of ERIT reduces to 
\begin{align}
J_{erit}(\rho) =\sum_{k=0}^{N-1} \Biggl(  \frac{1}{C_{ff}(\rho_0)+T_d C_{fb}} \left(T_d u_k^0 - C_{ff}(\rho)y_k^0 \right)\Biggr)^2,   \label{eq:JeritL}
\end{align}
which corresponds to $J_L(\rho)$ with 
$L=\frac{1}{C_{ff}(\rho_0)+T_d C_{fb}}$ (or possibly $L=\frac{1}{C_{ff}(\rho_0)+T_d C_{fb}}\frac{1}{z^m}$ to make $L$ proper). 
This result can be summarized as follows;
\begin{thm} \label{thm:ERIT}
Assume that $\frac{1}{C_{ff}(\rho_0)+T_d C_{fb}}$ is stable, and (A1) is satisfied. 
Then, ERIT gives a solution of 
\begin{align}
\minimize\ \|R(T(\rho)-T_d)\|^2, 
\end{align}
when $N\to \infty$. 
\end{thm}
The proof is almost trivial from Theorem~\ref{thm:main}, thus omitted. 
Theorem~\ref{thm:ERIT} suggests that ERIT is optimal when the weight function is exactly $R$.  
This implies that ERIT is optimal if we keep employing the same $r_k$. 

\section{Numerical example} \label{sec:sim}

This section gives a numerical example of the proposed optimal shaping filter, 
and shows its effectiveness. 

\subsection{General setting}
In this section, 
$P=\frac{z}{z^2-1.4z+0.98}$ is employed as a target system. 
The reference model and the feedback controller are 
set to $T_d = \frac{10^{-3}z^2}{(z-0.9)^3}$ and $C_{fb}=1$, respectively. 
The feedforward controller $C_{ff}(\rho)$ \rev{is parametrized as} 
\begin{align}
C_{ff}(\rho) = \frac{\rho_1}{z-\rho_2}, \rho=\begin{bmatrix}
\rho_1 & \rho_2
\end{bmatrix}^{\top}. 
\end{align}
This structure is too simple to represent $\frac{T_d}{P}$, thus an appropriate tuning of $\rho$ is 
critical to minimize $\|W(T(\rho)-T_d)\|^2$. 
The initial parameter $\rho_0=\begin{bmatrix}
0& 0
\end{bmatrix}^{\top}$, thus $C_{ff}(\rho_0))=0$. 
The random reference $r_k=R\varepsilon_k$ is used with $R=\frac{z}{z-0.4}$. 
The length of the experiment is  set to $N=800$. 
\rev{
In the following, 
the MATLAB command \texttt{fminsearch} is employed to minimize the cost function. 
}

\subsection{Case 1: $W=1$}

We first consider the case with $W=1$. 
\begin{figure}
\centering
\includegraphics[width=8cm]{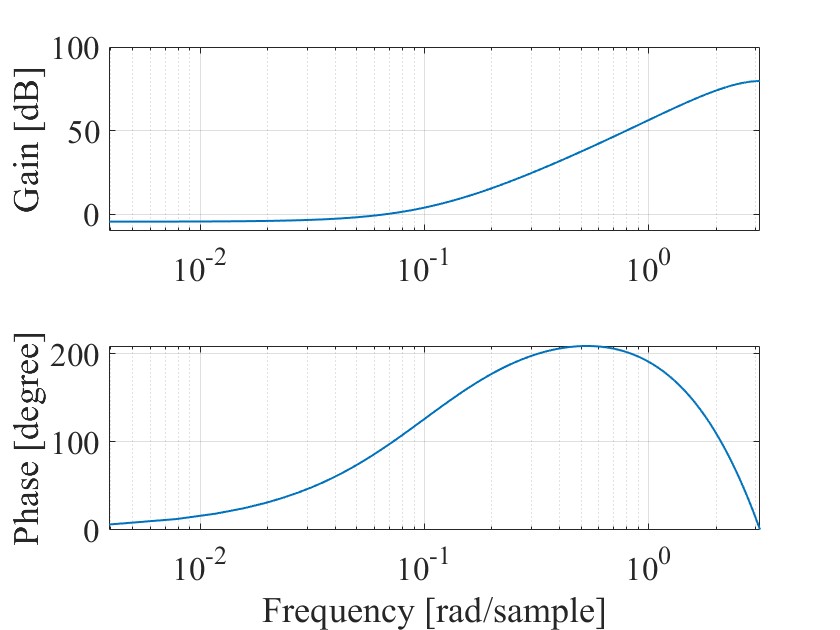}
\caption{Bode diagram of optimal shaping filter}
\label{fig:bodeL}
\end{figure}
Fig.~\ref{fig:bodeL} shows the Bode diagram of the proposed optimal shaping filter in this 
case. In this case, $m=1$ is employed to make $L$ proper. 
The horizontal axes show the frequency, and the vertical axes show the gain and the phase, respectively. 
In this case, the optimal shaping filter shows high-pass behavior. 
\begin{rem}
A high-pass shaping filter may increase the influence of measurement noise 
in practical applications. 
This problem has been recognized in the context of ERIT, and \cite{Fujimoto:2022} gives 
a way to reduce the influence of noise by signal projection. 
\rev{The effectiveness of the signal projection is shown in Sec.~\ref{sec:exp}. }
\end{rem}
\begin{figure}
\centering
\includegraphics[width=8cm]{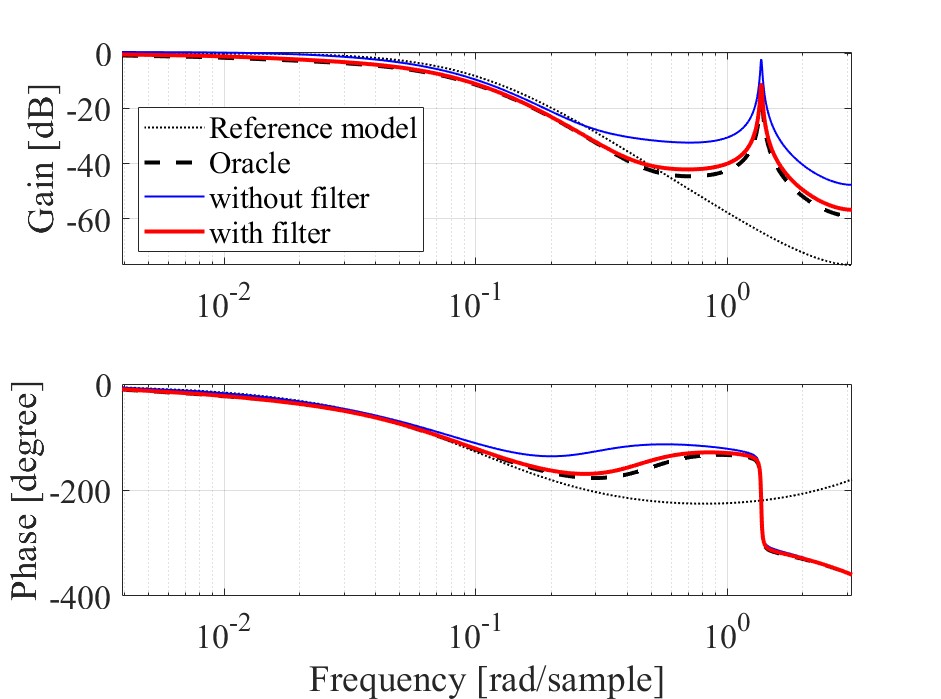}
\caption{Bode diagram of $T(\rho)$}
\label{fig:bodeT_ex1}
\end{figure}
Fig.~\ref{fig:bodeT_ex1} shows the bode diagram of $T(\rho)$ with three different $\rho$; 
the one with shaping filter (minimizer of $J_L(\rho)$), the one without shaping filter 
(minimizer of $J_0(\rho)$), and the one tuned with the true $P$ (minimizer of $J(\rho)$). 
The last one is called \textit{oracle case} in the following. 
In Fig.~\ref{fig:bodeT_ex1}, 
\rev{the dotted, the broken, the thin solid, and the thick solid lines} 
show the reference model, oracle case, the one without shaping filter, and the one with shaping filter, respectively. 
Since $C_{ff}(\rho)$ is too simple in this case, the oracle case is still different from 
the reference model. Recall that the shaping filter plays \rev{a} crucial role in such \rev{a} case. 
If the oracle perfectly matches the reference model, it means $C^*\in \mathcal{C}$ and 
the shaping filter is not required. 
In the low frequency range (lower than $0.2$ [rad/sample]), 
the case without a shaping filter shows closer behavior to the reference model than the others. 
In the high-frequency range, on the other hand, the oracle case and the case with a shaping filter 
become closer to the reference model. 
In particular, the shaping filter makes the result much closer to the oracle one. 

\begin{figure}
\centering
\includegraphics[width=8cm]{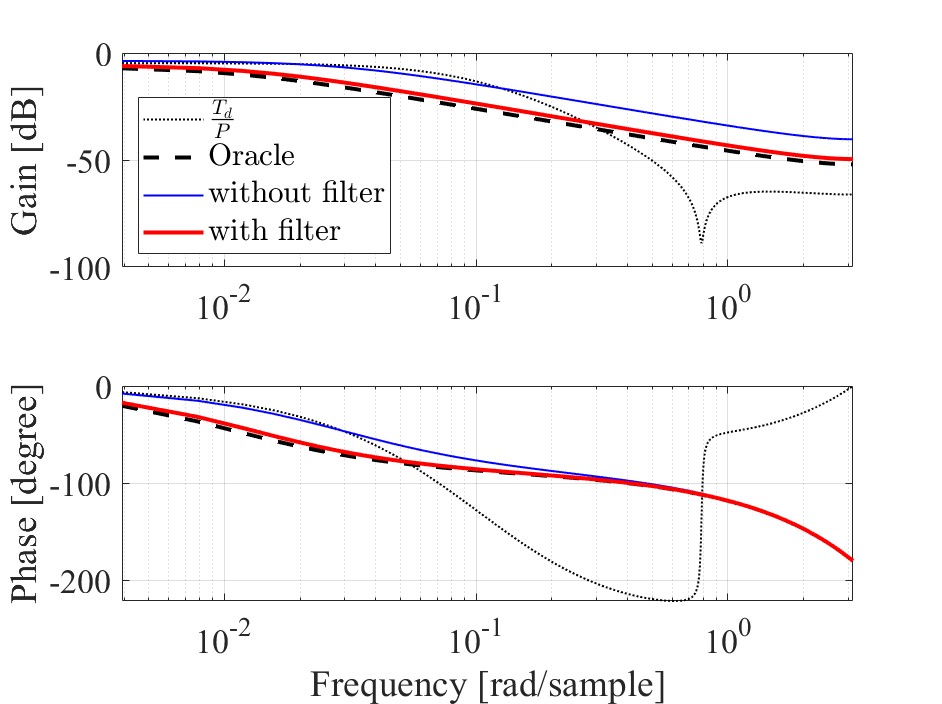}
\caption{Bode diagram of feedforward controllers}
\label{fig:bodeCff_ex1}
\end{figure}
Fig.~\ref{fig:bodeCff_ex1} shows the Bode diagram of feedforward controllers. 
\rev{the dotted, the broken, the thin solid, and the thick solid lines} 
show the optimal controller $\frac{T_d}{P}$, the oracle case, the one without shaping filter, and the one with shaping filter, respectively. 
Recall that the case without shaping filter can be understood as an identification of the optimal controller, 
thus the one without shaping filter becomes much closer to the optimal controller $\frac{T_d}{P}$. 
Also recall that similarity between $C_{ff}(\rho)$ and $C^*$ does not 
immediately mean the similarity between $T(\rho)$ and $T_d$. 
Hence, the case with shaping filter \rev{approximate $T_d$ well in Fig.~\ref{fig:bodeT_ex1} 
although it does not approximate the optimal controller well in Fig.~\ref{fig:bodeCff_ex1}}.

\subsection{Case 2: $W=R$}
This subsection shows the result with $W=R$, i.e., the case where minimizing $J_L(\rho)$ reduces to ERIT. 
Recall that ERIT is optimal when $W=R$, or intuitively, 
when the same $r_k$ is employed in the updated experiment. 
To demonstrate its effectiveness, we compare the 
output for the same $r_k$ with 
1) the oracle parameter, 2) 
the parameter tuned with the shaping filter, and 3) the one without the shaping filter. 

\begin{figure}
\centering
\includegraphics[width=8cm]{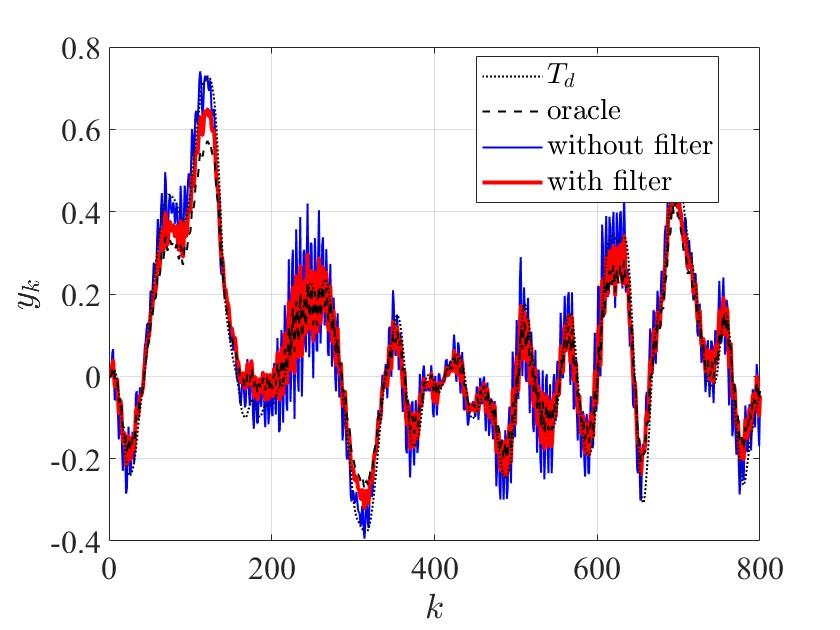}
\caption{Outputs with updated parameters}
\label{fig:outputs_ex2}
\end{figure}
\begin{figure}
\centering
\includegraphics[width=8cm]{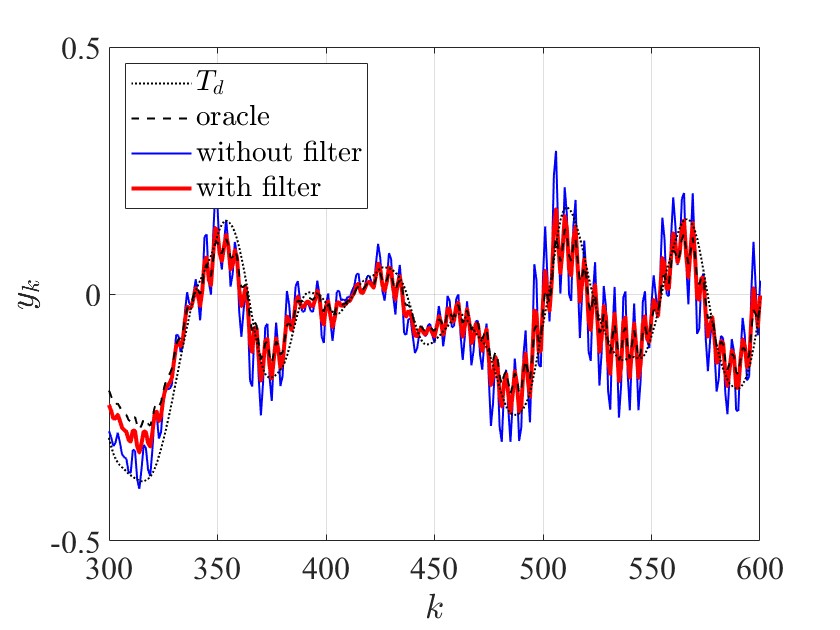}
\caption{Outputs with updated parameters ($300\leq k\leq 600$)}
\label{fig:outputs_ex2_large}
\end{figure}

Fig.~\ref{fig:outputs_ex2} shows the outputs with updated parameters, 
and Fig.~\ref{fig:outputs_ex2_large} shows its specific interval ($300\leq k \leq 600$). 
The horizontal \rev{axes show} the time step $k$, and the vertical \rev{axes show} $y_k$. 
\rev{The dotted and the broke lines show} $T_d r_k$ and $T(\rho^*) r_k$ where $\rho^*$ denotes the oracle parameter, 
respectively. 
The thin and thick solid lines show the result without and with the shaping filter, respectively. 
The output without the shaping filter shows oscillating behavior, and the one with the 
shaping filter 
much reduces such an oscillation. 
The squared error from $T_d r_k$ is 4.86 without filter, and 2.76 with filter. 
In this example, the squared error is reduced by about 43\% by employing the shaping filter. 

This result suggests that the shaping filter 
(or equivalently, using ERIT) 
plays an 
important role when considering the same reference is employed.  

\rev{
\section{Practical Experiment} \label{sec:exp}
To show the effectiveness of the proposed shaping filter, 
this section demonstrates a practical experiment with a Quanser Rotary Servo Base Unit. 
The input/output of this motor are voltage/angle in this experiment. 
The sampling period is $0.005$ [s], and we consider the case with 
step reference $r_k=\frac{\pi}{2}$ [rad]. 
The feedback controller $C_{fb}$ was set to $8+0.5\frac{1}{z-1}$, and the initial feedforward controller was zero. 
The reference model $T_d$ was set to $\frac{0.05^4 z^3}{(z-0.95)^4}$. 
We employed $W=R$ in this experiment, and compared the results with another experiment with an updated controller and the same step reference. 
To reduce the noise influence, we employed signal projection proposed in \cite{Fujimoto:2022}. 
The number of signal bases is set to 10, and the first order Butterworth filter with normalized cutoff frequency 0.02 is 
used to generate the bases. See \cite{Fujimoto:2022} for more information on this signal projection. 

\begin{figure}
\centering
\includegraphics[width=8cm]{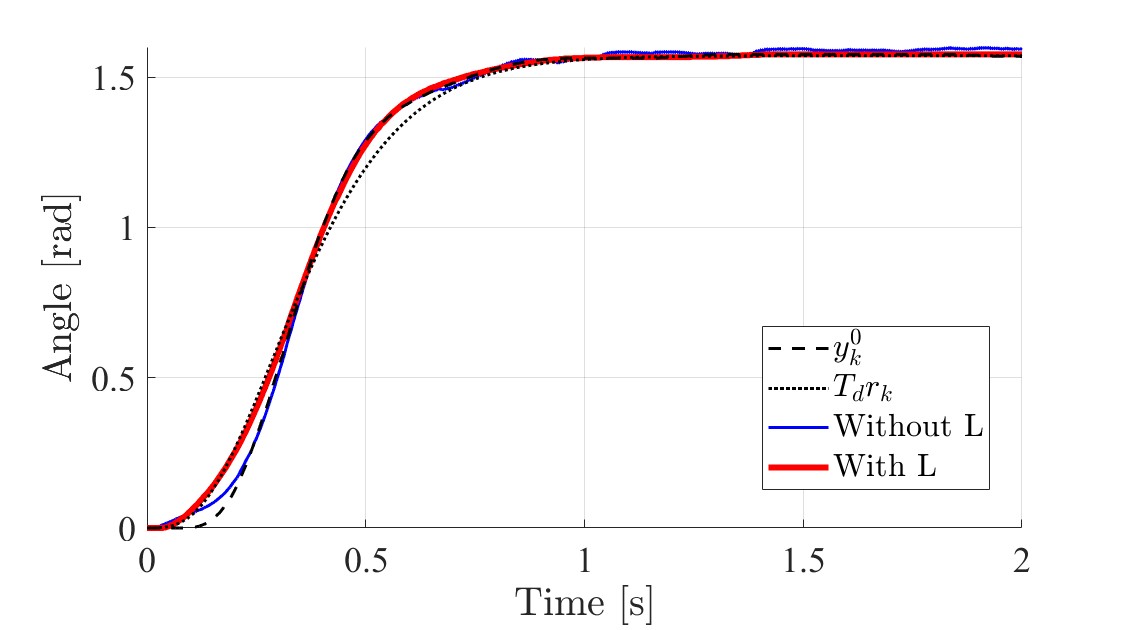}
\caption{Outputs with initial/updated controllers}
\label{fig:expresults}
\end{figure}
Fig.~\ref{fig:expresults} shows the outputs with/without a shaping filter. 
The horizontal axis shows time, and the vertical axis shows the angle. 
The dotted, broken, thin solid, and thick solid lines show 
the desired output, the initial output, the updated output without shaping filter, and the 
updated output with shaping filter, respectively. 
Fig.~\ref{fig:expresults} shows that the updated controller with the proposed shaping filter (thick solid line) well tracks 
the desired output (dotted line) especially around 0 to 0.3 [s]. 
It should also be noted that the updated controller without shaping filter (thin solid line) fails to track 
$T_d r_k$ in the same interval. 
The sum of squared error from $T_d r_k$ is 0.81, 0.76, and 0.26 for initial controller, the updated controller without $L$, and the one with $L$, respectively. 
This result shows that the proposed shaping filter works well for a practical system.


}

\section{Conclusion} \label{sec:conclusion}

This paper discusses the data-driven feedforward controller design 
 to achieve model matching. 
In particular, this paper focuses on the shaping filter design problem, and 
gives the optimal shaping filter. 
The proposed shaping filter is constructed from available information under reasonable assumptions; 
it requires an initial feedforward controller, reference model, feedback controller, and the spectrum of the 
initial experiment. 
A numerical example \rev{and a practical experiment are} shown to demonstrate the effectiveness of the proposed method. 
\rev{
This paper mainly discussed asymptotic properties. 
Discussing statistical properties with a limited amount of data is one of the future tasks. 
}

%

\bibliographystyle{IEEEtran}
\bibliography{myref}

\end{document}